\RequirePackage{fix-cm}

\documentclass[twocolumn,envcountsect]{svjour3}
\usepackage{amssymb,amsmath,latexsym}
\usepackage[varg]{pxfonts}
\usepackage[T1]{fontenc}
\usepackage{mathrsfs}
\usepackage{graphicx}

\smartqed  
\newcommand{\im}{\qopname\relax{no}{Im}}

\newcommand{\lBrack}{\lbrack\!\lbrack}
\newcommand{\rBrack}{\rbrack\!\rbrack}

\begin{document}

\title{Weighted finite automata with output\thanks{Research supported by Ministry of Education, Science and Technological Development, Republic of Serbia, Grant No. 174013.}}


\author{Jelena Ignjatovi\'c  \and Miroslav \'Ciri\'c \and Zorana Jan\v ci\'c}

\authorrunning{J. Ignjatovi\'c,  M. \'Ciri\'c, Z. Jan\v ci\'c}

\institute{J. Ignjatovi\'c \and  M. \'Ciri\'c \and Z. Jan\v ci\'c\at
              University of Ni\v s, Faculty of Sciences and Mathematics,\\ Department of Mathematics and Computer Science\\ Vi\v segradska 33, 18000 Ni\v s, Serbia \\
              Tel.: +38118224492\\
              Fax: +38118533014\\
              \email{jelena.ignjatovic@pmf.edu.rs,~miroslav.ciric@pmf.edu.rs,\break zoranajancic329@gmail.com}}

\date{Received: date / Accepted: date}

\maketitle

\begin{abstract} In this paper we prove the equivalence of sequential, Mealy-type and Moore-type weighted finite automata with output, with respect to various semantics which are defined here.

\hspace{5mm}

\end{abstract}

\section{Introduction}

Finite automata with output are a simple mathematical model of computation with
numerous applications in different areas.
Generally speaking, the main~role~of~an automaton with output is to transform finite sequences of~input symbols to finite sequences of output symbols, and its behavior is understood  as a function or relation between the sets of all input and output sequences.~The most simple among such~autom\-ata, the ordinary deter\-ministic finite automata with output,~have two~basic
models. Mealy-type automata simultaneously pass~into a new state and emit output,~and~the value of the output depends both on the current state and the current input, whereas Moore-type automata emit output just after the transition to the~next state, and the value of~the output depends solely on this new state.~Although different, these two models are equivalent, in~the sense~that any Mealy-type automaton can be converted
 to a Moore-type automaton with the~same~behavior, and vice versa.

When dealing with more complex types of automata, such as, for instance, fuzzy or weighted finite autom\-ata with output, things become more complicated.~Fuzzy
finite automata with output
have been studied by many authors who have considered several different models and
semantics.~Sequential fuzzy finite automata, where both transitions and
outputs are modeled by a single tran\-si\-tion-output function, have been investigated
in \cite{LP.06,MN.96,P.88,P.04b,PK.08,XQLF.07}.~It should be noted that
a sequential fuzzy finite automaton with an input alphabet $X$ and an output
alphabet $Y$ can be considered as a fuzzy automaton (i.e., fuzzy transition
system) with the input alphabet $X\times Y$ and without output.~On~the~other~hand, the articles \cite{LL.07,LP.06,LMQW.09,P.06,WQ.10} have dealt with Mealy-type and~Moore-type fuzzy finite automata, where transitions and outputs are modeled by separate transition and output functions, and the behavior of these automata has been defined in different ways.

All the mentioned models of automata are defined here in a more general context, for weighted finite~automata over a semiring.~Besides, the behavior of Mealy-type weighted finite
automata is defined  in three different~ways -- we distinguish the $n1$-semantics, the $1n$-se\-man\-tics  and the sequential semantics, whereas
for Moore-type automata we distinguish the $n1$-semantics and the $1n$-se\-man\-tics. In the framework of Mealy-type and Moore-type fuzzy finite automata the $n1$-semantics has been considered in \cite{CM.04,LL.07,LMQW.09,WQ.10}, the $n1$-semantics in \cite{LP.06,P.06},
and the sequential semantics in \cite{LP.06}.
The purpose of the paper is to study the equivalence between~the  mentioned types
of weighted finite automata with outputs, with respect to the  mentioned  semantics.~We~show that each Mealy-type weighted finite
automaton can be converted into a sequential weighted finite
automaton
equivalent w.r.t. the sequential semantics, each Moore-type weighted finite
automaton can~be~con\-verted into a sequential weighted finite
automaton
equivalent w.r.t. the $1n$-semantics, and vice versa, every Mealy-type weighted finite
automaton can be converted to a Moore-type weighted finite
automaton equivalent w.r.t.~both the $1n$-semantics
and $n1$-semantics, and each Moore-type weighted finite
automaton~can be~converted~to a Mealy-type weighted finite
automaton equivalent w.r.t. the $1n$-semantics. Moreover,~we determine certain
conditions under which a sequential weighted finite
automaton can be converted to a Mealy-type weighted finite
automaton equivalent w.r.t. the sequential seman\-tics. In all these cases
we also estimate the growth of~the number of states during the conversion.

Note that although different models of fuzzy~autom\-ata with output were studied in numerous papers, only the paper of Li and Pedrycz \cite{LP.06} discussed the equivalence of these models, and our work is a continuation of this research.

The paper is organized as follows. In Section 2 we recall basic notions and notation
concerning semirings and matrices over a semiring, and in Section 3~we~pre\-sent
definitions of sequential weighted automata and their behavior. Thereafter, in Sections 4 and 5 we define Mealy-type and Moore-type weighted finite automata,
 three different semantics for Mealy-type weighted automata and two semantics for Moore-type weighted~automata. Our main results are presented in Section 6, where we prove  the equivalence of sequential, Mealy-type and  Moore-type weighted finite automata with respect to various semantics. Finally, in Section 7 we consider crisp-deterministic Mealy-type and Moore-type weighted finite automata and show that all previously considered semantics coincide for such automata.

\section{Preliminaries}

Throughout this paper, $\Bbb N$ denotes the set of natural numbers (without zero),
$X^+$ and $X^*$ denote respectively the free semigroup and the free monoid over an alphabet $X$, and
$\varepsilon$ denotes the {\it empty word\/} in $X^*$.

A {\it semiring\/} is a structure  $(S,+,\cdot,0,1)$ consisting of~a set  $S$, two binary
operations $+$ and $\cdot$ on  $S$,  and two~con\-stants  $0, 1\in S$  such that the following is true:
\begin{itemize}\itemindent6pt\parskip0pt
\item[(i)] $(S,+,0)$ is a commutative monoid,
\item[(ii)] $(S,\cdot,1)$ is a monoid,
\item[(iii)] the distributivity laws $(r+s)\cdot t = r\cdot t+s\cdot t$ and $t\cdot (r+s) = t\cdot r+t\cdot s$
hold for every $r, s, t\in S$,
\item[(iv)] $0\cdot s = s\cdot 0 = 0$ for every $s\in S$.
\end{itemize}
As usual, we identify the structure $(S,+,\cdot,0,1)$ with its carrier
set $S$. A semiring $S$ is called {\it additively idempotent\/} if $s+s=s$, for every $s\in S$, or equivalently, if $1+1=1$. For $n\in \Bbb N$ and $s\in S$, the {\it $n$-th additive power\/} of $s$ is the element $ns=s+s+\ldots +s$ ($n$ times).

Let $P$ and $Q$ be sets.~We let $Q^P$ denote the set of all functions from $P$ to $Q$.~Next, let $S$ be a semiring and let $A$ be a finite non-empty set.~A mapping $\mu:A\times A\to S$ is called an {\it $A\times A$-matrix over\/} $S$, and a mapping $\nu :A\to S$ is called an {\it $A$-vector over\/} $S$.~If $S$ is a particular ordered set (e.g., the real unit interval $[0,1]$), then matrices are called {\it fuzzy relations\/}, and vectors are called {\it fuzzy subsets\/} in the literature.

Given matrices $\mu_1,\mu_2\in S^{A\times A}$ and vectors $\nu_1,\nu_2\in S^A$. Then we define the {\it matrix product}
$\mu_1\cdot \mu_2\in S^{A\times A}$, the {\it matrix-vector products\/} $\nu_1\cdot \mu_1\in S^A$ and $\mu_1\cdot \nu_1\in S^A$, and the {\it scalar product\/} $\nu_1\cdot \nu_2\in S$ as follows for every $a_1,a_2\in A$:
\begin{align*}
(\mu_1\cdot \mu_2)(a_1,a_2)&=\sum_{a\in A}\mu_1(a_1,a)\cdot \mu_2 (a,a_2),\\
(\nu_1\cdot \mu_1)(a_1)&=\sum_{a\in A} \nu_1(a)\cdot \mu_1 (a,a_1),\\
(\mu_1\cdot \nu_1)(a_1)&=\sum_{a\in A} \mu_1 (a_1,a)\cdot \nu_1(a),\\
\nu_1\cdot \nu_2&=\sum_{a\in A} \nu_1(a)\cdot \nu_2(a).
\end{align*}
Recall that the addition  of $S$ is commutative and that $A$ is non-empty; thus, the sums on the right-hand sides are well defined.~Moreover, since distributivity of the multiplication operation over the addition operation
holds, the matrix product and matrix-vector products are associative.~The
{\it Hadamard\/} ({\it pointwise\/}) {\it product\/} $\nu_1\odot \nu_2$ of vectors $\nu_1,\nu_2\in S^A$
is defined as follows for any $a\in A$:
\begin{equation*}
(\nu_1\odot\nu_2) (a)=\nu_1(a)\cdot \nu_2(a).
\end{equation*}
Given a vector $\nu\in S^A $, we define a matrix $D(\nu )\in S^{A\times A}$ as follows for every $a,b\in A$:
\begin{equation*}\label{eq:Dalpha}
D(\nu)(a,b)=\begin{cases}
\nu(a) & \text{if}\ a=b, \\
\hfil 0 & \text{otherwise}.
\end{cases}
\end{equation*}
For an arbitrary matrix $\mu\in S^{A\times A}$ and $a,b\in A$, we can easily verify that
\begin{equation}\label{eq:RDalpha}
\begin{aligned}
&(D(\nu)\cdot \mu)(a,b)=\nu(a)\cdot \mu(a,b), \\
&(\mu\cdot D(\alpha))(a,b)=\mu(a,b)\cdot \nu (b).
\end{aligned}
\end{equation}

\section{Sequential weighted automata}

All weighted automata that will be discussed throughout this paper will have finite sets of states, input and output alphabets.~Such automata are usually called weighted finite automata, but here we omit the adjec\-tive ``finite'' because it will entail.

A {\it sequential weighted automaton\/} over a semiring $S$ is a tuple ${\cal A}=(A,X,Y,\sigma^A,\mu^A)$, where $A$, $X$ and $Y$ are finite non-empty sets,~called
respectively the {\it set of states\/}, the {\it input alphabet\/}, and the {\it output alphabet\/}, $\sigma^A:A\to S$ is the {\it initial weight vector\/} and $\mu^A:A\times X\times Y\times A\to S$ is the {\it weighted transition-output function\/}.~The functions $\mu^A$ and $\sigma^A$ can be understood as follows.~When the weighted automa\-ton $\cal A$ is in a state $a\in A$ and it receives the input symbol $x\in X$, we can interpret $\mu^A(a,x,y,b)$ as the degree to which $\cal A$ moves into a state $b\in A$ and emits the output symbol $y\in Y$. On the other hand, we can interpret $\sigma^A(a)$ as the degree to which $a\in A$ is an initial state.~Without danger of confusion, in cases when we deal with a single sequential weighted automaton, we will omit~the superscript $A$ in $\sigma^A$ and $\mu^A$.

Let us note that the free monoid $(X\times Y)^*$ is isomorphic to the submonoid of $X^*\times Y^*$ consisting of all pairs $(u,v)\in X^*\times Y^*$ such that $|u|=|v|$, and we will identify these two monoids, as is commonly done~in~algebra. Thus, the identity in $(X\times Y)^*$ is identified with the identity $(\varepsilon,\varepsilon)$ of $X^*\times Y^*$, where $\varepsilon $ denotes the identity (empty word) both in $X^*$ and $Y^*$.

For any pair $(x,y)\in X\times Y$ we define $\mu_{x,y}:A\times A\to S$ by $\mu_{x,y}(a,b)=\mu (a,x,y,b)$, for all $a,b\in A$, and for any  $(u,v)\in (X\times Y)^*$  the {\it weighted transition-output matrix\/} (or the {\it weighted transi\-tion-out\-put relation\/}) $\mu_{u,v}:A\times A\to S$ is defined
as~fol\-lows: If $a,b\in A$, then
\begin{equation}\label{eq:mu.e}
\mu_{\varepsilon,\varepsilon}(a,b)=\begin{cases}\ 1, & \text{if}\ a=b, \\ \ 0, & \mbox{otherwise,}
\end{cases}
\end{equation}
and if $a,b\in A$, $(u,v)\in (X\times Y)^+$ and $(x,y)\in X\times Y$, then
\begin{equation}\label{eq:mu.x}
\mu_{ux,vy}(a,b)= \sum _{c\in A} \mu_{u,v}(a,c)\cdot \mu_{x,y}(c,b).
\end{equation}
It is easy to check that
\begin{equation}\label{eq:mu.uv}
\mu_{up,vq}(a,b)= \sum _{c\in A} \mu_{u,v}(a,c)\cdot \mu_{p,q}(c,b),
\end{equation}
i.e.,
$\mu_{up,vq}=\mu_{u,v}\cdot \mu_{p,q}$, for all $a,b\in A$ and $(u,v),(p,q)\in (X\times Y)^*$. Therefore, if $u=x_1\ldots x_n$ and $v=y_1\ldots y_n$, where $x_1,\ldots ,x_n\in X$ and $y_1,\ldots ,y_n\in Y$, then
\begin{align}\label{eq:mu.x1xn}
&\mu_{u,v}(a,b)= \sum _{(c_1,\ldots ,c_{n-1})\in A^{n-1}}
\mu_{x_1,y_1}(a,c_1)\cdot \mu_{x_2,y_2}(c_1,c_2)\,\cdot\\
&\hspace{50mm}\cdot\ldots \cdot \mu_{x_n,y_n}(c_{n-1},b),\notag
\end{align}
or, in the matrix form, $\mu_{u,v}=\mu_{x_1,y_1}\cdot \mu_{x_2,y_2}\cdot\ldots \cdot \mu_{x_n,y_n}$.

\begin{definition}
The {\it behavior\/} of a  sequential weighted auto\-maton ${\cal A}$ is the function $\lBrack {\cal A}\rBrack :(X\times Y)^*\to S$ defined~by
\begin{equation}\label{eq:beh.seq}
\lBrack {\cal A}\rBrack(\varepsilon,\varepsilon)=\sum_{a,b\in A}\sigma(a)\cdot \mu_{\varepsilon,\varepsilon}(a,b)=\sum_{a\in A}\sigma(a),
\end{equation}
and
\begin{align}\label{eq:beh.seq}
\lBrack {\cal A}\rBrack(u,v)&=\sum_{a,b\in A}\sigma(a)\cdot \mu_{u,v}(a,b)\\ &= \sum_{(a,a_1,\ldots,a_n)\in A^{n+1}} \sigma(a)\cdot \mu_{x_1,y_1}(a,a_1)\,\cdot
\notag \\
&\hspace{13mm}\cdot\mu_{x_2,y_2}(a_1,a_2)\cdot\ldots\cdot \mu_{x_n,y_n}(a_{n-1},a_n),\notag
\end{align}
for each $(u,v)\in (X\times Y)^+$, $u=x_1\ldots x_n$, $v=y_1\ldots y_n$, for some $n\in \Bbb N$, $x_1,\ldots ,x_n\in X$ and $y_1,\ldots ,y_n\in Y$.
In other words,
\begin{equation}\label{eq:beh.seq2}
\lBrack {\cal A}\rBrack(u,v)=\sigma\cdot \mu_{u,v}\cdot \tau,
\end{equation}
where $\tau:A\to S$ is given by $\tau(a)=1$, for any $a\in A$.
\end{definition}

\section{Mealy-type weighed  automata}\label{sec:Mealy}

A {\it Mealy-type weighted automaton\/} over a semiring $S$ is a tuple ${\cal A}=(A,X,Y,\sigma^A,\delta^A,\omega^A)$, where $A$, $X$, $Y$ and $\sigma^A$ are as in the definition of a sequential weighted automaton, $\delta^A:A\times X\times A\to S$ is the {\it weighted transition function\/}, and $\omega^A:A\times X\times Y\to S$ is the {\it weighted output~function\/}. The functions $\delta^A$
and $\omega^A$ can be understood as~follows. When the automaton $\cal A$ is in a state $a\in A$ and it receives the input symbol $x\in X$, we can interpret $\delta^A(a,x,b)$ as the degree to which $\cal A$ moves into a state $b\in A$, and $\omega^A(a,x,y)$ as the degree to which $\cal A$ emits the output symbol $y\in Y$. When we deal with a single Mealy-type weighted automaton, we omit the superscript $A$ in $\sigma^A$, $\delta^A$ and $\omega^A$.

For any $x\in X$ we define $\delta_x:A\times A\to S$  by $\delta_x(a,b)=\delta(a,x,b)$, for all $a,b\in A$, and for any $u\in X^*$ we define the {\it  weighted transition
matrix\/} (or {\it weighted transi\-tion relation\/}) $\delta_u:A\times A\to S$ as~fol\-lows: For any $a,b\in A$ we set
\begin{equation}\label{eq:delta.e}
\delta_{\varepsilon}(a,b)=\begin{cases}\ 1, & \text{if}\ a=b, \\ \ 0, & \mbox{otherwise,}
\end{cases}
\end{equation}
and if $a,b\in A$, $u\in X^*$ and $x\in X$, then
\begin{equation}\label{eq:delta.x}
\delta_{ux}(a,b)= \sum _{c\in A} \delta_u(a,c)\cdot \delta_x(c,b).
\end{equation}
It is easy to verify that
\begin{equation}\label{eq:delta.uv}
\delta_{uv}(a,b)= \sum _{c\in A} \delta_{u}(a,c)\cdot \delta_{v}(c,b),
\end{equation}
for all $a,b\in A$ and $u,v\in X^*$, i.e., $\delta_{uv}=\delta_u\cdot \delta_v$.
Hence, if $u=x_1\cdots x_n$, for some $n\in \Bbb N$ and $x_1,\ldots ,x_n\in X$,
then
\begin{align}\label{eq:delta.x1xn}
&\delta_u(a,b)= \\
&\hspace{5mm}=\sum _{(c_1,\ldots ,c_{n-1})\in A^{n-1}}
\delta_{x_1}(a,c_1)\cdot \delta_{x_2}(c_1,c_2)\, \cdot \ldots \cdot \delta_{x_n}(c_{n-1},b),\notag
\end{align}
i.e., $\delta_u=\delta_{x_1}\cdot \delta_{x_2}\cdot\ldots\cdot \delta_{x_n}$.

Next, define a vector $\omega_{\varepsilon,\varepsilon}:A\to S$ by $\omega_{\varepsilon,\varepsilon}(a)=1$,~for any $a\in A$, and for any pair $(x,y)\in X\times Y$  define~a vector
 $\omega_{x,y}:A\to S$ by $\omega_{x,y}(a)=\omega(a,x,y)$, for every $a\in A$. For an arbitrary $(u,v)\in (X\times Y)^+$ a vector $\omega_{u,v}:A\to
S$  can be defined in three ways.

\begin{definition}[$\mathbf{1n}$-semantics]
For any $(x,y)\in X\times Y$ and $(u,v)\in (X\times Y)^+$ we set
\begin{equation}\label{eq:1n-sem.Mealy}
\omega_{xu,yv}=D(\omega_{x,y})\cdot \delta_x\cdot \omega_{u,v}.
\end{equation}
In other words, for any $n\in \Bbb N$, $x_1,\ldots ,x_n\in X$ and $y_1,\ldots ,y_n\in Y$, we have that
\begin{align}\label{eq:1n-sem2.Mealy}
\omega_{x_1\ldots x_n,y_1\ldots y_n}&=D(\omega_{x_1,y_1})\cdot \delta_{x_1}\cdot D(\omega_{x_2,y_2})\cdot \delta_{x_2}\cdot\ldots\cdot\\
&\hspace{20mm}\cdot D(\omega_{x_{n-1},y_{n-1}})\cdot \delta_{x_{n-1}}\cdot \omega_{x_n,y_n},\notag
\end{align}
i.e., for every $a\in A$ the following is true
\begin{align}\label{eq:1n-sem3.Mealy}
&\omega_{x_1\cdots x_n,y_1\ldots y_n}(a)=
\sum_{(a_1,\ldots,a_{n-1})\in A^{n-1}} \omega_{x_1,y_1}(a)\cdot \delta_{x_1}(a,a_1)\,\cdot\\
&\hspace{12mm}\cdot \omega_{x_2,y_2}(a_1)\cdot \delta_{x_2}(a_1,a_2)\cdot\ldots\cdot\omega_{x_{n-1},y_{n-1}}(a_{n-2})\,\cdot \notag \\
&\hspace{24mm}\cdot  \delta_{x_{n-1}}(a_{n-2},a_{n-1})\cdot\omega_{x_n,y_n}(a_{n-1}).\notag
\end{align}
The {\it $1n$-behavior\/} of ${\cal A}$ is the function $\lBrack {\cal A}\rBrack_{1n}:(X\times Y)^*\to S$ defined by
\begin{equation}\label{eq:1n-beh.Mealy.e}
\lBrack {\cal A}\rBrack_{1n}(\varepsilon,\varepsilon)=\sigma\cdot \omega_{\varepsilon,\varepsilon}=\sum_{a\in A}\sigma(a)
\end{equation}
and
\begin{align}\label{eq:1n-beh.Mealy}
&\lBrack {\cal A}\rBrack_{1n}(u,v)=\sigma\cdot \omega_{u,v}=\sum_{a\in A}\sigma(a)\cdot \omega_{u,v}(a) \\
&\hspace{5mm}= \sum_{(a,a_1,\ldots,a_{n-1})\in A^{n}} \sigma(a)\cdot \omega_{x_1,y_1}(a)\cdot \delta_{x_1}(a,a_1)\,\cdot\notag \\
&\hspace{15mm}\cdot\omega_{x_2,y_2}(a_1)\cdot \delta_{x_2}(a_1,a_2)\cdot\ldots
\cdot \omega_{x_{n-1},y_{n-1}}(a_{n-2})\,\cdot\notag \\
&\hspace{25mm}\cdot \delta_{x_{n-1}}(a_{n-2},a_{n-1})\cdot \omega_{x_n,y_n}(a_{n-1}),\notag
\end{align}
for each $(u,v)\in (X\times Y)^+$, $u=x_1\ldots x_n$, $v=y_1\ldots y_n$, for some $n\in \Bbb N$, $x_1,\ldots ,x_n\in X$ and $y_1,\ldots ,y_n\in Y$.
\end{definition}

\begin{definition}[$\mathbf{n1}$-semantics]
For any $(x,y)\in X\times Y$ and $(u,v)\in (X\times Y)^+$ we set
\begin{equation}\label{eq:n1-sem.Mealy}
\omega_{ux,vy}=D(\omega_{u,v})\cdot \delta_u\cdot \omega_{x,y}.
\end{equation}
In other words, for each $n\in \Bbb N$, $x_1,\ldots ,x_n\in X$, and $y_1,\ldots ,y_n\in Y$ we have that
\begin{align}\label{eq:n1-sem2.Mealy}
&\omega_{x_1\cdots x_n,y_1\ldots y_n}=
\omega_{x_1,y_1}\odot (\delta_{x_1}\cdot \omega_{x_2,y_2})\,\odot \\
&\hspace{15mm}\odot(\delta_{x_1x_2}\cdot \omega_{x_3,y_3})\odot\cdots\odot
(\delta_{x_1\ldots x_{n-1}}\cdot \omega_{x_n,y_n}),\notag
\end{align}
or equivalently, for every $a\in A$ we have
\begin{align}\label{eq:n1-sem3.Mealy}
&\omega_{x_1\cdots x_n,y_1\ldots y_n}(a)=\\
&\hspace{5mm}=
\sum_{(a_1,\ldots,a_{n-1})\in A^{n-1}} \omega_{x_1,y_1}(a)\cdot \delta_{x_1}(a,a_1)\cdot\omega_{x_2,y_2}(a_1)\,\cdot\notag
\\
&\hspace{15mm} \cdot \delta_{x_1x_2}(a,a_2)\cdot \omega_{x_3,y_3}(a_2)\cdot\ldots\cdot  \notag\\
&\hspace{30mm}\cdot \delta_{x_1\ldots x_{n-1}}(a,a_{n-1})\cdot \omega_{x_n,y_n}(a_{n-1}).\notag
\end{align}
In this case, the {\it $n1$-behavior\/} of ${\cal A}$ is defined as the function $\lBrack {\cal A}\rBrack_{n1}:(X\times Y)^*\to S$ given by
\begin{equation}\label{eq:n1-beh.Mealy.e}
\lBrack {\cal A}\rBrack_{n1}(\varepsilon,\varepsilon)=\sigma\cdot \omega_{\varepsilon,\varepsilon}=\sum_{a\in A}\sigma(a)
\end{equation}
and
\begin{align}\label{eq:n1-beh.Mealy}
&\lBrack {\cal A}\rBrack_{n1}(u,v)=\sigma\cdot \omega_{u,v}=\sum_{a\in A}\sigma(a)\cdot \omega_{u,v}(a) \\
&\hspace{5mm}= \sum_{(a,a_1,\ldots,a_{n-1})\in A^{n}} \sigma(a)\cdot \omega_{x_1,y_1}(a)\cdot \delta_{x_1}(a,a_1)\cdot \omega_{x_2,y_2}(a_1)\,\cdot \notag \\
&\hspace{15mm}\cdot\delta_{x_1x_2}(a,a_2)\cdot \omega_{x_3,y_3}(a_2)\cdot\ldots\cdot \notag \\
&\hspace{25mm}\cdot \delta_{x_1\ldots x_{n-1}}(a,a_{n-1})\cdot \omega_{x_n,y_n}(a_{n-1}),\notag
\end{align}
for each $(u,v)\in (X\times Y)^+$, $u=x_1\ldots x_n$, $v=y_1\ldots y_n$, for some $n\in \Bbb N$, $x_1,\ldots ,x_n\in X$ and $y_1,\ldots ,y_n\in Y$.
\end{definition}

\begin{definition}[Sequential semantics]
The {\it $s$-behavior\/} of ${\cal A}$ is the function $\lBrack {\cal A}\rBrack_{s}:(X\times Y)^*\to S$ defined by
\begin{equation}\label{eq:s-beh.Mealy.e}
\lBrack {\cal A}\rBrack_{s}(\varepsilon,\varepsilon)=\sigma\cdot \omega_{\varepsilon,\varepsilon}=\sum_{a\in A}\sigma(a)
\end{equation}
and
\begin{align}\label{eq:s-beh.Mealy}
&\lBrack {\cal A}\rBrack_{s}(u,v)= \sum_{(a,a_1,\ldots,a_{n})\in A^{n+1}} \sigma(a)\cdot \omega_{x_1,y_1}(a)\cdot \delta_{x_1}(a,a_1)\,\cdot \\
&\hspace{6mm}\cdot\omega_{x_2,y_2}(a_1)\cdot \delta_{x_2}(a_1,a_2)\cdot\ldots\cdot \omega_{x_n,y_n}(a_{n-1})\cdot\delta_{x_{n}}(a_{n-1},a_{n}),\notag
\end{align}
for each $(u,v)\in (X\times Y)^+$, $u=x_1\ldots x_n$, $v=y_1\ldots y_n$, for some $n\in \Bbb N$, $x_1,\ldots ,x_n\in X$ and $y_1,\ldots ,y_n\in Y$.
\end{definition}

If we define a function $\mu :A\times X\times Y\times A\to S$ by
\begin{equation}\label{eq:od-mu}
\mu(a,x,y,b)=\omega(a,x,y)\cdot \delta(a,x,b),
\end{equation}
for all $a,b\in A$, $x\in X$ and $y\in Y$, i.e., if we set
\begin{equation}\label{eq:od-mu2}
\mu_{x,y}(a,b)=\omega_{x,y}(a)\cdot \delta_{x}(a,b),
\end{equation}
for all $a,b\in A$ and $(x,y)\in X\times Y$, we obtain a sequential weighted automaton ${\cal A}'=(A,X,Y,\sigma,\mu )$ such that $\lBrack {\cal A}'\rBrack =\lBrack {\cal A}\rBrack_{s}$. For this reason this semantics is called sequential.

Let us note that $\mu_{x,y}=D(\omega_{x,y})\cdot\delta_x$, for all $x\in X$ and $y\in Y$, and therefore,
\begin{align}\label{eq:mu.uv}
&\mu_{x_1\ldots x_n,y_1\ldots y_n}= \\
&\hspace{5mm}=D(\omega_{x_1,y_1})\cdot \delta_{x_1}\cdot D(\omega_{x_2,y_2})\cdot \delta_{x_2}\cdot \ldots \cdot D(\omega_{x_n,y_n})\cdot \delta_{x_n},\notag
\end{align}
for any $n\in \Bbb N$ and $x_1,\ldots ,x_n\in X$, $y_1,\ldots ,y_n\in Y$.

\begin{example}\rm
Let $S=([0,1],\lor,\land,0,1)$ be the G\"odel semi\-ring, and  ${\cal A}=(A,X,Y,\sigma ,\delta ,\omega )$  a Mealy-type weight\-ed automaton over $S$ with  $|A|=2$, $X=\{0\}$, $Y=\{0,1\}$, and
\begin{equation*}\label{eq:ex.Mealy}
\begin{aligned}
&\sigma=[\,1\ \ 0\,], \quad \delta_0=\begin{bmatrix} 0.7 & 0.5 \\ \ 0 & \ 0.8 \end{bmatrix},  \\
&\omega_{0,0}=[\,0.6\ \ 0.4\,], \quad \omega_{0,1}=[\, 0.2\ \ 0.7\,].
\end{aligned}
\end{equation*}
It is easy to check that
\begin{equation*}
\begin{aligned}
\lBrack {\cal A}\rBrack_{1n}(000,010)&=\lBrack {\cal A}\rBrack_{s}(000,010)=0.4
\\
&\ne 0.5=\lBrack {\cal A}\rBrack_{n1}(000,010).
\end{aligned}
\end{equation*}
Therefore, both the $1n$-semantics and the sequential semantics differ from the $n1$-semantics.
\end{example}

\begin{example}\rm
Again, let $S$ be the G\"odel semiring, and let ${\cal A}=(A,X,Y,\sigma ,\delta ,\omega )$ be a Mealy-type weighted automaton over $S$ with  $|A|=2$, $X=\{0,1\}$, $Y=\{0\}$, and
\begin{equation*}\label{eq:ex.Mealy2}
\begin{aligned}
&\sigma=[\,1\ \ 0\,], \quad \delta_0=\displaystyle\begin{bmatrix} 0.7 & 0.5 \\ \ 0 & \ 0.8 \end{bmatrix}, \quad \delta_1=\begin{bmatrix} 0.3 & \ 1 \\ 0.2 & \ 0 \end{bmatrix},\\
&\omega_{0,0}=[\,0.6\ \ 0.4\,], \quad \omega_{1,0}=[\, 0.2\ \ 0.7\,].
\end{aligned}
\end{equation*}
Then
\[
\lBrack {\cal A}\rBrack_{1n}(01,00)=0.5\ne 0.2=\lBrack {\cal A}\rBrack_{s}(01,00),
\]
and hence, the $1n$-semantics and the sequential semantics may  also be different.
\end{example}

\section{Moore-type weighted automata}\label{sec:Moore}

A {\it Moore-type weighted automaton\/} over a semiring $S$ is~a tuple ${\cal A}=(A,X,Y,\sigma^A,\delta^A,\omega^A)$, where everything is the same as in the definition of a Mealy-type weighted~auto\-maton except the weighted output function, for which we assume that $\omega^A:A\times Y\to S$.

Here, we define $\omega_{\varepsilon,\varepsilon}:A\to S$ by $\omega_{\varepsilon,\varepsilon}(a)=1$,~for each $a\in A$, and for any $(x,y)\in X\times Y$ we define $\omega_{x,y}:A\to L$ by $\omega_{x,y}=\delta_x\cdot \omega_y$. For $(u,v)\in (X\times Y)^+$ we can define a vector $\omega_{u,v}:A\to S$ in two ways.

\begin{definition}[$\mathbf{1n}$-semantics]
For each $(x,y)\in X\times Y$ and $(u,v)\in (X\times Y)^+$ we set
\begin{equation}\label{eq:1n-sem.Moore}
\omega_{xu,yv}=\delta_x\cdot D(\omega_y)\cdot \omega_{u,v}.
\end{equation}
In other words, for each $n\in \Bbb N$, $x_1,\ldots ,x_n\in X$ and $y_1,\ldots ,y_n\in Y$, we have that
\begin{equation}\label{eq:1n-sem2.Moore}
\omega_{x_1\ldots x_n,y_1\ldots y_n}=\delta_{x_1}\cdot D(\omega_{y_1})\cdot \delta_{x_2}\cdot D(\omega_{y_2})\cdot \ldots\cdot \delta_{x_{n}}\cdot \omega_{y_n},
\end{equation}
i.e., for every $a\in A$ we have
\begin{align}\label{eq:1n-sem3.Moore}
&\omega_{x_1\cdots x_n,y_1\ldots y_n}(a)=
\sum_{(a_1,\ldots,a_{n})\in A^{n}} \delta_{x_1}(a,a_1)\cdot\omega_{y_1}(a_1)\,\cdot \\
&\hspace{10mm}\cdot \delta_{x_2}(a_1,a_2)\cdot \omega_{y_2}(a_2)\cdot \ldots
\cdot \delta_{x_{n}}(a_{n-1},a_{n})\cdot \omega_{y_n}(a_{n}).\notag
\end{align}
The {\it $1n$-behavior\/} of ${\cal A}$ is the function $\lBrack {\cal A}\rBrack_{1n}:(X\times Y)^*\to L$ defined by
\begin{equation}\label{eq:1n-beh.Moore.e}
\lBrack {\cal A}\rBrack_{1n}(\varepsilon,\varepsilon)=\sigma\cdot \omega_{\varepsilon,\varepsilon}=\sum_{a\in A}\sigma(a)
\end{equation}
and
\begin{align}\label{eq:1n-beh.Moore}
&\lBrack {\cal A}\rBrack_{1n}(u,v)=\sigma\cdot \omega_{u,v}=\sum_{a\in A}\sigma(a)\cdot \omega_{u,v}(a) \\
&\hspace{6mm}= \sum_{(a,a_1,\ldots,a_{n})\in A^{n+1}} \sigma(a)\cdot \delta_{x_1}(a,a_1)\cdot\omega_{y_1}(a_1)\,\cdot \notag \\ &\hspace{12mm}\cdot \delta_{x_2}(a_1,a_2)\cdot \omega_{y_2}(a_2)\cdot \ldots
\cdot \delta_{x_{n}}(a_{n-1},a_{n})\cdot \omega_{y_n}(a_{n}),\notag
\end{align}
for each $(u,v)\in (X\times Y)^+$, $u=x_1\ldots x_n$, $v=y_1\ldots y_n$, for some $n\in \Bbb N$, $x_1,\ldots ,x_n\in X$ and $y_1,\ldots ,y_n\in Y$.
\end{definition}

Note that a slightly different definition of $1n$-seman\-tics for Moore-type fuzzy finite automata was given by Li and Pedrycz in \cite{LP.06}.

\begin{definition}[$\mathbf{n1}$-semantics]
For each $(x,y)\in X\times Y$ and $(u,v)\in (X\times Y)^+$ we set
\begin{equation}\label{eq:n1-sem.Moore}
\omega_{ux,vy}=D(\omega_{u,v})\cdot \delta_{ux}\cdot \omega_{y}.
\end{equation}
In other words, for each $n\in \Bbb N$, $x_1,\ldots ,x_n\in X$, and $y_1,\ldots ,y_n\in Y$ we have that
\begin{align}\label{eq:n1-sem2.Moore}
&\omega_{x_1\cdots x_n,y_1\ldots y_n}=\\
&\hspace{10mm}=(\delta_{x_1}\cdot \omega_{y_1})\odot (\delta_{x_1x_2}\cdot \omega_{y_2})\odot\cdots\odot
(\delta_{x_1\ldots x_{n}}\cdot \omega_{y_n}),\notag
\end{align}
which means that
\begin{align}\label{eq:n1-sem3.Moore}
&\omega_{x_1\cdots x_n,y_1\ldots y_n}(a)= \sum_{(a_1,\ldots,a_{n})\in A^{n}} \delta_{x_1}(a,a_1)\cdot\omega_{y_1}(a_1)\,\cdot \\
&\hspace{10mm}\cdot\delta_{x_1x_2}(a,a_2)\cdot\omega_{y_2}(a_2)\cdot \ldots
\cdot \delta_{x_1\ldots x_{n}}(a,a_{n})\cdot \omega_{y_n}(a_{n}),\notag
\end{align}
for every $a\in A$.
Now, the {\it $n1$-behavior\/} of ${\cal A}$ is defined as the function $\lBrack {\cal A}\rBrack_{n1}:(X\times Y)^*\to S$ given by
\begin{equation}\label{eq:n1-beh.Moore.e}
\lBrack {\cal A}\rBrack_{n1}(\varepsilon,\varepsilon)=\sigma\cdot \omega_{\varepsilon,\varepsilon}=\sum_{a\in A}\sigma(a)
\end{equation}
and
\begin{align}\label{eq:1n-beh.Moore}
&\lBrack {\cal A}\rBrack_{n1}(u,v)=\sigma\cdot \omega_{u,v}=\sum_{a\in A}\sigma(a)\cdot \omega_{u,v}(a) \\
&\hspace{6mm}= \sum_{(a,a_1,\ldots,a_{n})\in A^{n+1}} \sigma(a)\cdot \delta_{x_1}(a,a_1)\cdot\omega_{y_1}(a_1)\,\cdot \notag\\
&\hspace{12mm}\cdot \delta_{x_1x_2}(a,a_2)\cdot \omega_{y_2}(a_2)\cdot \ldots
\cdot \delta_{x_1\ldots x_{n}}(a,a_{n})\cdot \omega_{y_n}(a_{n}),\notag
\end{align}
for every $(u,v)\in (X\times Y)^+$, $u=x_1\ldots x_n$, $v=y_1\ldots y_n$, for some $n\in \Bbb N$, $x_1,\ldots ,x_n\in X$ and $y_1,\ldots ,y_n\in Y$.
\end{definition}

\section{Equivalence of sequential, Moore-type and Mealy-type weighted automata}

Two weighted finite automata with output of any type (sequential, Mealy-type
or Moore type) are \emph{equivalent} if they have equal behaviors
(with respect to the considered semantics).
In this section  we prove theorems on  the equivalence of sequential, Mealy-type and Moore-type weight\-ed finite automata with respect to various semantics.

First we prove that every Mealy-type weighted automaton $\cal A$ can be converted into a sequential weighted automaton which is equivalent to $\cal A$ with respect to sequential semantics on $\cal A$.

\begin{theorem}\label{th:Mealy-seq}
For any Mealy-type weighted automaton ${\cal A}=(A,X,Y,\sigma^A,\delta^A,\omega^A)$ there exists a sequential weighted auto\-maton ${\cal B}=(B,X,Y,\sigma^B,\mu^B)$ such that
\[
\lBrack {\cal A}\rBrack_s =\lBrack {\cal B}\rBrack.
\]
In addition, $\cal B$ can be chosen so that $|{\cal B}|\leqslant |{\cal A}|$.
\end{theorem}

\begin{proof}
Set $B=A$ and define $\mu^B:B\times X\times Y\times B\to S$ and $\sigma^B:B\to S$ by $\sigma^B=\sigma^A$ and
\begin{equation*}
\mu (a,x,y,b)=\omega (a,x,y)\cdot \delta(a,x,b),
\end{equation*}
for all $a,b\in A$, $x\in X$ and $y\in Y$. Then it is easy to check that $\lBrack {\cal A}\rBrack_s =\lBrack {\cal B}\rBrack$.
\qed
\end{proof}

Next, we show that every Moore-type weighted automaton $\cal A$ can be converted into a sequential weighted automaton which is equivalent to $\cal A$ with respect to $1n$-semantics on $\cal A$.

\begin{theorem}\label{th:Moore-seq}
For any Moore-type weighted automaton ${\cal A}=(A,X,Y,\sigma^A,\delta^A,\omega^A)$ there exists a sequential weighted auto\-maton ${\cal B}=(B,X,Y,\sigma^B,\mu^B)$ such that
\[
\lBrack {\cal A}\rBrack_{1n} =\lBrack {\cal B}\rBrack.
\]
In addition, $\cal B$ can be chosen so that $|{\cal B}|\leqslant |{\cal A}|$.
\end{theorem}

\begin{proof}
Set $B=A$ and define $\mu^B:B\times X\times Y\times B\to S$ and $\sigma^B:B\to S$ by $\sigma^B=\sigma^A$ and
\begin{equation*}
\mu (a,x,y,b)= \delta(a,x,b)\cdot\omega (b,y),
\end{equation*}
for all $a,b\in A$, $x\in X$ and $y\in Y$. Then $\lBrack {\cal A}\rBrack_{1n} =\lBrack {\cal B}\rBrack$.
\qed
\end{proof}

On the other hand, the next theorem shows that any sequen\-tial weighted automaton $\cal A$ can be converted to a Moore-type weighted automaton $\cal B$ which is equivalent to $\cal A$ with respect to $1n$-semantics on $\cal B$.

\begin{theorem}\label{th:seq-Moore}
For any sequential weighted automaton ${\cal A}=(A,X,Y,\sigma^A,\mu^A)$ there exists a Moore-type weighted automa\-ton ${\cal B}=(B,X,Y,\sigma^B,\delta^B,\omega^B)$ such that
\[
\lBrack {\cal A}\rBrack =\lBrack {\cal B}\rBrack_{1n}.
\]
In addition, $\cal B$ can be chosen so that $|{\cal B}|\leqslant |{\cal A}|\cdot |Y|$.
\end{theorem}

\begin{proof}
Set $B=A\times Y$ and fix an arbitrary $y_0\in Y$. Define $\sigma^B:B\to S$, $\delta^B:B\times X\times B\to S$ and $\omega^B:B\times Y\to S$ as follows: For $b,b_1,b_2\in B$, $x\in X$ and $y\in Y$ we set
\begin{align*}
&\sigma^B(b)=\begin{cases}
\sigma^A(a) & \text{if}\ \ b=(a,y_0),\ \text{for some}\ a\in A, \\
\hfil 0 & \text{otherwise},
\end{cases} \\
&\delta^B(b_1,x,b_2)=\mu^A(a_1,x,y_2,a_2), \qquad \text{if}\ \ b_1=(a_1,y_1),\\
&\hspace{15mm} b_2=(a_2,y_2),\ \text{for some}\ a_1,a_2\in A,\ y_1,y_2\in Y,\notag
\\
&\omega^B(b,y)=\begin{cases}
\ 1 & \text{if}\ \ b=(a,y),\ \text{for some}\ a\in A, \\
\ 0 & \text{otherwise}.
\end{cases}
\end{align*}
Then ${\cal B}=(B,X,Y,\sigma^B,\delta^B,\omega^B)$ is a Moore-type weight\-ed automa\-ton. We are going to prove that $\cal A$ is equivalent to $\cal B$ with respect to the $1n$-semantics of $\cal B$.

Take an arbitrary $(u,v)\in (X\times Y)^+$, where $u=x_1\ldots x_n$, $v=y_1\ldots y_n$, for $n\in \Bbb N$, $x_1,\ldots ,x_n\in X$, $y_1,\ldots ,y_n\in Y$. Consider any $(b_0,b_1,\ldots ,b_n)\in B^{n+1}$ and the product
\begin{align}\label{eq:prod.S.Mo}
&\sigma^B(b_0)\cdot  \delta_{x_1}^B(b_0,b_1)\cdot\omega_{y_1}^B(b_1)\cdot \delta_{x_2}^B(b_1,b_2)\,\cdot\omega_{y_2}^B(b_2)\cdot  \notag\\
&\hspace{10mm}\cdot\ldots\cdot\delta_{x_n}^B(b_{n-1},b_{n})\cdot \omega_{y_n}^B(b_{n}).
\end{align}
If for each $i\in \{0,1,\ldots ,n\}$ we have that
\begin{equation}\label{eq:b.i}
b_{i}=(a_{i},y_i),\ \ \text{for some}\ a_{i}\in A,
\end{equation}
then
\begin{align*}
&\sigma^B(b_0)=\sigma^A(a_0), \ \ \omega^B_{y_i}(b_{i})=1, \\
&\delta^B_{x_i}(b_{i-1},b_i)=\mu^A_{x_i,y_i}(a_{i-1},a_i),
\end{align*}
for each $i\in \{1,\ldots ,n\}$, and the product (\ref{eq:prod.S.Mo}) becomes
\[
\sigma^A(a_0)\cdot \mu_{x_1,y_1}^A(a_0,a_1)\cdot\mu_{x_2,y_2}^A(a_1,a_2)\cdot\ldots\cdot \mu_{x_n,y_n}^A(a_{n-1},a_n).
\]
Otherwise, if there exists $i\in \{0,1,\ldots ,n\}$ such that $b_i$ can not be
written in the form (\ref{eq:b.i}), then we have that $\omega^B_{y_i}(b_{i})=0$ (for $i\geqslant 1$) or $\sigma^B(b_0)=0$ (for $i=0$), and the whole product (\ref{eq:prod.S.Mo}) is equal to $0$.

Note also that there is a one-to-one correspondence between all $(n+1)$-tuples $(a_0,a_1,\ldots ,a_n)\in A^{n+1}$ and all $(n+1)$-tuples $(b_0,b_1,\ldots ,b_n)\in B^{n+1}$ satisfying (\ref{eq:b.i}), which implies
\begin{align*}
&\lBrack {\cal B}\rBrack_{1n}(u,v)= \sum_{(b_0,b_1,\ldots ,b_n)\in B^{n+1}}\sigma^B(b_0)\cdot \delta_{x_1}^B(b_0,b_1)\cdot \omega_{y_1}^B(b_1)\cdot \\
&\hspace{5mm}\cdot \delta_{x_2}^B(b_1,b_2)\cdot\omega_{y_2}^B(b_2)\cdot\ldots\cdot\delta_{x_{n}}^B(b_{n-1},b_{n})\cdot \omega_{y_n}^B(b_{n}) \\
&\hspace{10mm}=\sum_{(a_0,a_1,\ldots ,a_n)\in A^{n+1}}\sigma^A(a_0)\cdot \mu_{x_1,y_1}^A(a_0,a_1)\cdot\mu_{x_2,y_2}^A(a_1,a_2)\,\cdot
\\
&\hspace{20mm}\cdot\ldots\cdot \mu_{x_n,y_n}^A(a_{n-1},a_n)=\lBrack {\cal A}\rBrack (u,v),
\end{align*}
and hence, $\lBrack {\cal B}\rBrack_{1n}=\lBrack {\cal A}\rBrack$.
Clearly, $|{\cal B}|\leqslant |{\cal A}|\cdot |Y|$.
\qed
\end{proof}

Then we show that any Mealy-type weighted autom\-aton ${\cal A}$ can be converted into a Moore-type weighted automaton ${\cal B}$ such that $\cal A$ and $\cal B$ are equivalent both with respect to $1n$-semantics and $n1$-semantics.

\begin{theorem}\label{th:Mealy-Moore}
For every Mealy-type weighted automaton ${\cal A}=(A,X,Y,\sigma^A,\delta^A,\omega^A)$ there exists a Moore-type weighted auto\-maton ${\cal B}=(B,X,Y,\sigma^B,\delta^B,\omega^B)$ such that
\[
\lBrack {\cal A}\rBrack_{1n}=\lBrack {\cal B}\rBrack_{1n}\quad\text{and}\quad
\lBrack {\cal A}\rBrack_{n1}=\lBrack {\cal B}\rBrack_{n1}.
\]
In addition, $\cal B$ can be chosen so that $|{\cal B}|\leqslant |{\cal A}|\cdot (|X|+1)$.
\end{theorem}

\begin{proof}
Set $B=A\cup A\times X$.~Let us define $\sigma^B:B\to S$, $\delta^B:B\times X\times B\to S$ and $\omega^B:B\times Y\to S$ as follows: For $b,b_1,b_2\in B$, $x\in X$ and $y\in Y$ we set
\begin{align*}
&\sigma^B(b)=\begin{cases}
\sigma^A(a) & \text{if}\ \ b=a\in A, \\
\hfil 0 & \text{otherwise},
\end{cases} \\
&\delta^B(b_1,x,b_2)=\begin{cases}
\hfil 1 & \text{if}\ \ b_1=a\in A,\\
&\hspace{3mm}\ b_2=(a,x)\in A\times X , \\
\delta^A(a_1,x_1,a_2) & \text{if}\ \ b_1=(a_1,x_1)\in A\times X,\\
&\hspace{3mm}\ b_2=(a_2,x)\in A\times X,\\
\hfil 0 & \text{otherwise},
\end{cases} \\
&\omega^B(b,y)=\begin{cases}
\omega^A(a,x,y) & \text{if}\ \ b=(a,x)\in A\times X, \\
\hfil 0 & \text{otherwise}.
\end{cases}
\end{align*}
Then ${\cal B}=(B,X,Y,\sigma^B,\delta^B,\omega^B)$ is a Moore-type weighted automaton.

Take an arbitrary $(u,v)\in (X\times Y)^+$, where $u=x_1\ldots x_n$, $v=y_1\ldots y_n$, for $n\in \Bbb N$, $x_1,\ldots ,x_n\in X$, $y_1,\ldots ,y_n\in Y$. Consider any $(b_0,b_1,\ldots ,b_n)\in B^{n+1}$ and the product
\begin{align}\label{eq:prod.Me.Mo}
&\sigma^B(b_0)\cdot  \delta_{x_1}^B(b_0,b_1)\cdot\omega_{y_1}^B(b_1)\cdot \delta_{x_2}^B(b_1,b_2)\,\cdot\omega_{y_2}^B(b_2)\cdot  \notag\\
&\hspace{10mm}\cdot\ldots\cdot\delta_{x_{n}}^B(b_{n-1},b_{n})\cdot \omega_{y_n}^B(b_{n}).
\end{align}
Suppose that
\begin{equation}\label{eq:b.0.Me.Mo}
b_{0}=a_{0},\ \ \text{for some}\ a_{0}\in A,
\end{equation}
and for any $i\in \{1,\ldots ,n\}$ suppose that
\begin{equation}\label{eq:b.i.Me.Mo}
b_{i}=(a_{i-1},x_i),\ \ \text{for some}\ a_{i-1}\in A.
\end{equation}
Then
\begin{align}
&\sigma^B(b_0)=\sigma^A(a_0), \label{eq:Me.Mo.s}\\
&\omega^B_{y_i}(b_{i})=\omega^A_{x_i,y_i}(a_{i-1}),\
\text{for}\ i\in \{1,\ldots,n\}, \label{eq:Me.Mo.o}\\
&\delta^B_{x_1}(b_{0},b_1)=1, \label{eq:Me.Mo.d1}\\
&\delta^B_{x_i}(b_{i-1},b_i)=\delta^A_{x_{i-1}}(a_{i-2},a_{i-1}), \
\text{for}\ i\in \{2,\ldots,n\},\label{eq:Me.Mo.di}
\end{align}
and the product (\ref{eq:prod.Me.Mo}) becomes
\begin{align*}
&\sigma^A(a_0)\cdot \omega_{x_1,y_1}^A(a_0)\cdot \delta_{x_1}^A(a_0,a_1)\cdot\omega_{x_2,y_2}^A(a_1)\cdot \delta_{x_2}^A(a_1,a_2)\,\cdot \notag \\
&\hspace{15mm}\cdot\ldots\cdot \delta^A_{x_{n-1}}(a_{n-2},a_{n-1})\cdot \omega_{x_n,y_n}^A(a_{n-1}) .
\end{align*}

On the other hand, if $b_0\in A\times X$ or if $b_i$ can not~be written in the
form (\ref{eq:b.i.Me.Mo}), for some $i\in \{1,\ldots ,n\}$, i.e., if $b_i\in
A$ or $b_i=(a,x)\in A\times X$ such that $x\ne x_i$, then $\sigma^B(b_0)=0$
or  $\delta^B_{x_i}(b_{i-1},b_i)=0$,
and in both cases the whole product (\ref{eq:prod.Me.Mo}) is equal to~$0$.

Since to each $(n+1)$-tuple $(a_0,\ldots ,a_n)\in A^{n+1}$~cor\-re\-sponds
exactly one $(n+1)$-tuple $(b_0,\ldots ,b_n)\in B^{n+1}$~satis\-fying
(\ref{eq:b.0.Me.Mo})  and (\ref{eq:b.i.Me.Mo}), we have that
\begin{align*}
&\lBrack {\cal B}\rBrack_{1n}(u,v)= \sum_{(b_0,\ldots ,b_n)\in B^{n+1}}\sigma^B(b_0)\cdot \delta_{x_1}^B(b_0,b_1)\cdot \omega_{y_1}^B(b_1)\cdot \\
&\hspace{5mm}\cdot \delta_{x_2}^B(b_1,b_2)\cdot\omega_{y_2}^B(b_2)\cdot\ldots\cdot\delta_{x_{n}}^B(b_{n-1},b_{n})\cdot \omega_{y_n}^B(b_{n}) \\
&\hspace{5mm}=\sum_{(a_0,\ldots ,a_n)\in A^{n+1}}\sigma^A(a_0)\cdot \omega_{x_1,y_1}^A(a_0)\cdot \delta_{x_1}^A(a_0,a_1)\cdot \omega_{x_2,y_2}^A(a_1)\,\cdot \\
&\hspace{15mm}\cdot \delta_{x_2}^A(a_1,a_2)\,\,\cdot\ldots\cdot \delta^A_{x_{n-1}}(a_{n-2},a_{n-1})\cdot \omega_{x_n,y_n}^A(a_{n-1})
\\
&\hspace{5mm}=\lBrack {\cal A}\rBrack_{1n}(u,v),
\end{align*}
and hence, $\lBrack {\cal B}\rBrack_{1n}=\lBrack {\cal A}\rBrack_{1n}$.

Next, we prove that $\lBrack {\cal B}\rBrack_{n1}=\lBrack {\cal A}\rBrack_{n1}$.
Consider again an arbitrary $(b_0,b_1,\ldots ,b_n)\in B^{n+1}$ and the product
\begin{align}\label{eq:prod.Me.Mo.n1}
&\sigma^B(b_0)\cdot  \delta_{x_1}^B(b_0,b_1)\cdot\omega_{y_1}^B(b_1)\cdot \delta_{x_1x_2}^B(b_{0},b_2)\,\cdot\omega_{y_2}^B(b_2)\cdot  \notag\\
&\hspace{10mm}\cdot\ldots\cdot\delta_{x_1\ldots x_{n}}^B(b_{0},b_{n})\cdot \omega_{y_n}^B(b_{n}).
\end{align}
Suppose again that (\ref{eq:b.0.Me.Mo}) and (\ref{eq:b.i.Me.Mo}) hold. Then
we have that (\ref{eq:Me.Mo.s}), (\ref{eq:Me.Mo.o}) and (\ref{eq:Me.Mo.d1})
also hold. Now, take an arbitrary $j\in \{2,\ldots ,n\}$ and $(b_1',\ldots
,b_{j-1}')\in B^{j-1}$, and consider the product
\begin{equation}\label{eq:prod.b.primes}
\delta_{x_1}^B(b_0,b_1')\cdot \delta_{x_2}^B(b_1',b_2')\cdot \ldots \cdot
\delta_{x_j}^B(b_{j-1}',b_j).
\end{equation}
If
\begin{equation}\label{eq:b.0p.Me.Mo}
b_{1}'=(a_{0},x_1),
\end{equation}
and if for any $k\in \{2,\ldots ,j-1\}$ we have that
\begin{equation}\label{eq:b.ip.Me.Mo}
b_{k}'=(a_{k-1}',x_k),\ \ \text{for some}\ a_{k-1}'\in A,
\end{equation}
then
\begin{align*}
&\delta_{x_1}^B(b_0,b_1')=1,   \\
&\delta_{x_2}^B(b_{1}',b_2')=\delta_{x_1}^A(a_{0},a_{1}'), \\
&\delta_{x_k}^B(b_{k-1}',b_k')=\delta_{x_{k-1}}^A(a_{k-2}',a_{k-1}'),\ \ \text{for}\ k\in \{3,\ldots,j-1\}, \\
&\delta_{x_j}^B(b_{j-1}',b_j)=\delta_{x_{j-1}}^A(a_{j-2}',a_{j-1}),
\end{align*}
and consequently, the product (\ref{eq:prod.b.primes}) becomes
\begin{align*}
\delta_{x_1}^A(a_{0},a_{1}')\cdot \delta_{x_2}^A(a_{1}',a_{2}')\cdot\ldots\cdot
\delta_{x_{j-1}}^A(a_{j-2}',a_{j-1}).
\end{align*}
Otherwise, if $b_{1}'\ne (a_{0},x_1)$ or there is $k\in \{2,\ldots, j-1\}$
such that (\ref{eq:b.i.Me.Mo}) does not hold, then $\delta_{x_k}^B(b_{k-1}',b_k')=0$
and the product (\ref{eq:prod.b.primes}) is also equal to $0$. Therefore
\begin{align*}
&\delta_{x_1\ldots x_j}^B (b_0,b_j)= \\
&\hspace{2mm}=\sum_{(b_1',\ldots,b_{j-1}')\in B^{j-1}}\delta_{x_1}^B(b_0,b_1')\cdot
\delta_{x_2}^B(b_1',b_2')\cdot \ldots  \cdot \delta_{x_{j}}^B(b_{j-1}',b_j)
\\
&\hspace{2mm}=\sum_{(a_1',\ldots,a_{j-2}')\in A^{j-2}}\delta_{x_1}^A(a_0,a_1')\cdot
\delta_{x_2}^A(a_1',a_2')\cdot \ldots  \cdot \delta_{x_{j-1}}^A(a_{j-2}',a_{j-1})
\\
&\hspace{2mm}=\delta_{x_1\ldots x_{j-1}}^A (a_0,a_{j-1}),
\end{align*}
for any $j\in \{2,\ldots ,n\}$, which means that the product (\ref{eq:prod.Me.Mo.n1})
becomes
\begin{align*}
&\sigma^A(a_0)\cdot \omega_{x_1,y_1}^A(a_0)\cdot \delta_{x_1}^A(a_0,a_1)\cdot
\omega_{x_2,y_2}^A(a_1)\cdot \delta_{x_1x_2}^A(a_0,a_2)\,\cdot \\
&\hspace{20mm}\cdot \ldots\cdot \delta_{x_1\ldots x_{n-1}}^A(a_0,a_{n-1})\cdot
\omega_{x_n,y_n}^A(a_{n-1}).
\end{align*}
Next, if $b_0\in A\times X$ or there exists $i\in \{1,\ldots ,n\}$ such that
$b_i=(a,x)\in A\times X$ with $x\ne x_i$, then $\delta_{x_i}^B(b,b_i)=0$, for
any $b\in B$, whence $\delta_{x_1\ldots x_i}^B(b_0,b_i)=0$, which implies
that the product (\ref{eq:prod.Me.Mo.n1}) is equal to $0$.
Now, we conclude that
\begin{align*}
&\lBrack {\cal B}\rBrack_{n1}(u,v)= \sum_{(b_0,\ldots ,b_n)\in B^{n+1}}
\sigma^B(b_0)\cdot  \delta_{x_1}^B(b_0,b_1)\cdot\omega_{y_1}^B(b_1)\,\cdot
\\
&\hspace{5mm}\cdot \delta_{x_1x_2}^B(b_{0},b_2)\,\cdot\omega_{y_2}^B(b_2)\cdot  \ldots\cdot\delta_{x_1\ldots x_{n}}^B(b_{0},b_{n})\cdot \omega_{y_n}^B(b_{n})
 \\
&\hspace{16mm}\,=\sum_{(a_0,\ldots ,a_{n-1})\in A^{n}}\sigma^A(a_0)\cdot \omega_{x_1,y_1}^A(a_0)\cdot \delta_{x_1}^A(a_0,a_1)\,\cdot \\
&\hspace{23mm}\cdot\omega_{x_2,y_2}^A(a_1)\cdot \delta_{x_1x_2}^A(a_0,a_2)\,\cdot\omega_{x_3,y_3}^A(a_2)\,\cdot \\
&\hspace{30mm}\cdot \ldots\cdot \delta_{x_1\ldots x_{n-1}}^A(a_0,a_{n-1})\cdot
\omega_{x_n,y_n}^A(a_{n-1}) \\
&\hspace{16mm}\,=\lBrack {\cal A}\rBrack_{n1}(u,v).
\end{align*}
Therefore, $\lBrack {\cal B}\rBrack_{n1}=\lBrack {\cal A}\rBrack_{n1}$.
\qed
\end{proof}

We also prove that any Moore-type weighted autom\-aton ${\cal A}$ can be converted into a Mealy-type weighted automaton ${\cal B}$ such that $\cal A$ and $\cal B$ are equivalent with respect to $1n$-semantics.

\begin{theorem}\label{th:Moore-Mealy}
For every Moore-type weighted automaton ${\cal A}=(A,X,Y,\sigma^A,\delta^A,\omega^A)$ there is a Mealy-type weighted~au\-to\-maton ${\cal B}=(B,X,Y,\sigma^B,\delta^B,\omega^B)$ such that
\[
\lBrack {\cal A}\rBrack_{1n}=\lBrack {\cal B}\rBrack_{1n}.
\]
In addition, $\cal B$ can be chosen so that $|{\cal B}|\leqslant |{\cal A}|^2$.
\end{theorem}

\begin{proof}
Let $B=A\times A$ and let $\sigma^B:B\to S$, $\delta^B:B\times X\times B\to S$ and $\omega^B:B\times X\times Y\to S$ be defined as follows: For $b,b_1,b_2\in B$, $x\in X$ and $y\in Y$ we set
\begin{align*}
&\sigma^B(b)=\sigma^A(a), \qquad\ \  \text{if}\ \ b=(a,a'), \text{for some}\ a,a'\in A,\\
&\delta^B(b_1,x,b_2)=\begin{cases}
\ \ 1 & \text{if}\ \ b_1=(a_1,a_2),\, b_2=(a_2,a_3),\\
&\hspace{3mm}\ \text{for some}\ a_1,a_2,a_3\in A, \\
\ \ 0 & \text{otherwise},
\end{cases} \\
&\omega^B(b,x,y)=\delta_x^A(a_1,a_2)\cdot \omega_y(a_2),\quad \text{if}\ \ b=(a_1,a_2), \\
&\hspace{48mm}\ \text{for some}\ a_1,a_2\in A,
\end{align*}
Then ${\cal B}=(B,X,Y,\sigma^B,\delta^B,\omega^B)$ is a Mealy-type weighted automaton. To prove that $\lBrack {\cal A}\rBrack_{1n}=\lBrack {\cal B}\rBrack_{1n}$ take an arbitrary $(u,v)\in (X\times Y)^+$, where $u=x_1\ldots x_n$, $v=y_1\ldots y_n$, for some $n\in \Bbb N$, $x_1,\ldots ,x_n\in X$, $y_1,\ldots ,y_n\in Y$. Consider any $(b_0,b_1,\ldots ,b_{n-1})\in B^{n}$ and the product
\begin{align}\label{eq:prod.Mo.Me}
&\sigma^B(b_0)\cdot \omega_{x_1,y_1}^B(b_0)\cdot \delta_{x_1}^B(b_0,b_1)\cdot\omega_{x_2,y_2}^B(b_1)\,\cdot  \\
&\hspace{30mm}\cdot\ldots\cdot \delta_{x_2}^B(b_1,b_2)\cdot \omega_{x_n,y_n}^B(b_{n-1}).
\notag
\end{align}
Suppose that there is $(a_0,a_1,\ldots ,a_n)\in A^{n+1}$ such that
\begin{equation}\label{eq:b.seq}
b_{i-1}=(a_{i-1},a_i), \quad \text{for each}\ i\in \{1,\ldots ,n\}.
\end{equation}
Then we have that
\begin{align*}
&\sigma^B(b_0)=\sigma^A(a_0), \\
&\delta_{x_i}^B(b_{i-1},b_i)=1, \quad\text{for each}\ i\in \{1,\ldots ,n-1\},\\
&\omega_{x_i,y_i}^B(b_{i-1})=\delta_{x_1}^A(a_{i-1},a_i)\cdot \omega_{y_i}^A(a_i),
\end{align*}
and consequently, the product (\ref{eq:prod.Mo.Me}) becomes
\begin{align*}
&\sigma^A(a_0)\cdot \delta_{x_1}^A(a_0,a_1)\cdot \omega_{y_1}^A(a_1)\,\cdot\delta_{x_2}^A(a_1,a_2)\cdot \\
&\hspace{25mm}\cdot  \omega_{y_2}^A(a_2)\cdot\ldots\cdot  \delta_{x_n}^A(a_{n-1},a_n)\cdot \omega_{y_n}^A(a_n).
\end{align*}
On the other hand, if there is  $i\in \{1,\ldots ,n-1\}$ such that $b_{i-1}=(a_1',a_2')$, $b_{i}=(a_1'',a_2'')$ and $a_2'\ne a_1''$, then we obtain that $\delta_{x_i}^B(b_{i-1},b_{i})=0$,
and the product (\ref{eq:prod.Mo.Me})
is equal to $0$.

Since for any $(n+1)$-tuple $(a_0,a_1,\ldots ,a_n)\in A^{n+1}$ there exists a unique $n$-tuple $(b_0,b_1,\ldots ,b_{n-1})\in B^n$ such that (\ref{eq:b.seq})
holds, we have that
\begin{align*}
&\lBrack {\cal B}\rBrack_{1n}(u,v)=\sum_{(b_0,b_1,\ldots ,b_{n-1})\in B^n}\sigma^B(b_0)\cdot \omega_{x_1,y_1}^B(b_0)\cdot\delta_{x_1}^B(b_0,b_1)\,\cdot \\
&\hspace{8mm}\cdot\omega_{x_2,y_2}^B(b_1)\cdot\ldots\cdot \delta_{x_2}^B(b_1,b_2)\cdot \omega_{x_n,y_n}^B(b_{n-1})\\
&\hspace{16mm}\,=\sum_{(a_0,a_1,\ldots ,a_n)\in A^{n+1}}\sigma^A(a_0)\cdot \delta_{x_1}^A(a_0,a_1)\cdot \omega_{y_1}^A(a_1)\,\cdot \\
&\hspace{8mm}\cdot\delta_{x_2}^A(a_1,a_2)\cdot  \omega_{y_2}^A(a_2)\cdot\ldots\cdot  \delta_{x_n}^A(a_{n-1},a_n)\cdot \omega_{y_n}^A(a_n) \\
 &\hspace{16mm}\,=\lBrack {\cal A}\rBrack_{1n}(u,v),
\end{align*}
and hence, $\lBrack {\cal B}\rBrack_{1n}=\lBrack {\cal A}\rBrack_{1n}$.
\qed
\end{proof}

Finally, we show that under certain conditions~a~sequential weighted finite automaton $\cal A$ can be converted to a Mealy-type weighted finite automaton $\cal B$ equivalent to $\cal A$ with respect to the sequential semantics on~$\cal B$.

\begin{theorem}\label{th:seq-Mealy}
Given a sequential weighted automaton ${\cal A}=(A,X,Y,\sigma^A,\mu^A)$. If there exists $p\in \Bbb N$ such that $(pk)\,s=s$, for any $s\in \im (\mu^A )$, where $k=|X|\cdot |Y|$, then there is~a~Mealy-type weight\-ed automa\-ton ${\cal B}=(B,X,Y,\sigma^B,\delta^B,\omega^B)$ such that
\[
\lBrack {\cal A}\rBrack =\lBrack {\cal B}\rBrack_{s}.
\]
In addition, $\cal B$ can be chosen so that $|{\cal B}|\leqslant |{\cal A}|\cdot |X|\cdot |Y|$.
\end{theorem}

\begin{proof}
Set $B=A\times X\times Y$. Let us define $\sigma^B:B\to S$, $\delta^B:B\times X\times B\to S$ and $\omega^B:B\times X\times Y\to S$ as follows: For $b,b_1,b_2\in B$, $x\in X$ and $y\in Y$ we set
\begin{align*}
&\sigma^B(b)=p\,\sigma^A(a), \qquad\ \  \text{if}\ \ b=(a,x_1,y_1), \text{for some}\ a\in A,\\
&\hspace{32mm} x_1\in X\ \text{and}\ y_1\in Y, \\
&\delta^B(b_1,x,b_2)=\mu^A(a_1,x,y_1,a_2), \qquad \text{if}\ \ b_1=(a_1,x_1,y_1),\\
&\hspace{30mm} b_2=(a_2,x_2,y_2),\ \text{for some}\ a_1,a_2\in A,\\
&\hspace{30mm} x_1,x_2\in X,\ y_1,y_2\in Y,\notag
\\
&\omega^B(b,x,y)=\begin{cases}
\ 1 & \text{if}\ \ b=(a,x,y),\ \text{for some}\ a\in A, \\
\ 0 & \text{otherwise}.
\end{cases}
\end{align*}
Then ${\cal B}=(B,X,Y,\sigma^B,\delta^B,\omega^B)$ is a Mealy-type weight\-ed automa\-ton. We will show that $\cal A$ is equivalent to $\cal B$ with respect to the sequential semantics of $\cal B$.

Take an arbitrary $(u,v)\in (X\times Y)^+$, where $u=x_1\ldots x_n$, $v=y_1\ldots y_n$, for $n\in \Bbb N$, $x_1,\ldots ,x_n\in X$, $y_1,\ldots ,y_n\in Y$. Consider any $(b_0,b_1,\ldots ,b_n)\in B^{n+1}$ and the product
\begin{align}\label{eq:prod.S.Me}
&\sigma^B(b_0)\cdot \omega_{x_1,y_1}^B(b_0)\cdot \delta_{x_1}^B(b_0,b_1)\cdot\omega_{x_2,y_2}^B(b_1)\cdot \delta_{x_2}^B(b_1,b_2)\,\cdot  \notag\\
&\hspace{10mm}\cdot\ldots\cdot \omega_{x_n,y_n}^B(b_{n-1})\cdot\delta_{x_{n}}^B(b_{n-1},b_{n}).
\end{align}
If for each $i\in \{1,\ldots ,n\}$ we have that
\begin{equation}\label{eq:b.i-1}
b_{i-1}=(a_{i-1},x_i,y_i),\ \ \text{for some}\ a_{i-1}\in A,
\end{equation}
and if
\begin{equation}\label{eq:b.n}
b_{n}=(a_{n},x,y),\ \ \text{for some}\ a_n\in A,\ x\in X\ \text{and}\ y\in Y,
\end{equation}
then
\begin{align*}
&\sigma^B(b_0)=p\,\sigma^A(a_0), \ \ \omega^B_{x_i,y_i}(b_{i-1})=1, \\
&\delta^B_{x_i}(b_{i-1},b_i)=\mu^A_{x_i,y_i}(a_{i-1},a_i),
\end{align*}
for each $i\in \{1,\ldots ,n\}$, and the product (\ref{eq:prod.S.Me}) becomes
\begin{equation}\label{eq:prod.a}
\sigma^A(a_0)\cdot \mu_{x_1,y_1}^A(a_0,a_1)\cdot\mu_{x_2,y_2}^A(a_1,a_2)\cdot\ldots\cdot \mu_{x_n,y_n}^A(a_{n-1},a_n).
\end{equation}
Otherwise, if there is $i\in \{1,\ldots ,n\}$ such that (\ref{eq:b.i-1}) does not hold, then $\omega^B_{x_i,y_i}(b_{i-1})=0$, and the whole product (\ref{eq:prod.S.Me}) is equal to $0$.

For any $a_n\in A$ there are $k$ elements $b_n\in B$ satisfying (\ref{eq:b.n}), and thus, for any $(n+1)$-tuple $(a_0,\ldots ,a_n)\in A^{n+1}$ there are $k$ $(n+1)$-tuples $(b_0,\ldots ,b_n)\in B^{n+1}$ which satisfy (\ref{eq:b.i-1}) and (\ref{eq:b.n}). Consequently,
\begin{align*}
&\lBrack {\cal B}\rBrack_{s}(u,v)= \sum_{(b_0,b_1,\ldots ,b_n)\in B^{n+1}}\sigma^B(b_0)\cdot \omega_{x_1,y_1}^B(b_0)\cdot \delta_{x_1}^B(b_0,b_1)\cdot \\
&\hspace{5mm}\cdot\omega_{x_2,y_2}^B(b_1)\cdot \delta_{x_2}^B(b_1,b_2)\cdot\ldots\cdot \omega_{x_n,y_n}^B(b_{n-1})\cdot\delta_{x_{n}}^B(b_{n-1},b_{n}) \\
&\hspace{15mm}=\sum_{(a_0,a_1,\ldots ,a_n)\in A^{n+1}}[p\,\sigma^A(a_0)]\cdot \mu_{x_1,y_1}^A(a_0,a_1)\,\cdot\\
&\hspace{25mm}\cdot\mu_{x_2,y_2}^A(a_1,a_2)\cdot\ldots\cdot [k\,\mu_{x_n,y_n}^A(a_{n-1},a_n)]\\
&\hspace{15mm}=\sum_{(a_0,a_1,\ldots ,a_n)\in A^{n+1}}\sigma^A(a_0)\cdot \mu_{x_1,y_1}^A(a_0,a_1)\,\cdot\\
&\hspace{25mm}\cdot\mu_{x_2,y_2}^A(a_1,a_2)\cdot\ldots\cdot [(pk)\,\mu_{x_n,y_n}^A(a_{n-1},a_n)]\\
&\hspace{15mm}=\sum_{(a_0,a_1,\ldots ,a_n)\in A^{n+1}}\sigma^A(a_0)\cdot \mu_{x_1,y_1}^A(a_0,a_1)\,\cdot\\
&\hspace{25mm}\cdot\mu_{x_2,y_2}^A(a_1,a_2)\cdot\ldots\cdot \mu_{x_n,y_n}^A(a_{n-1},a_n)\\
&\hspace{15mm}=\lBrack {\cal A}\rBrack (u,v),
\end{align*}
and hence, $\lBrack {\cal B}\rBrack_{s}=\lBrack {\cal A}\rBrack$.
Clearly, $|{\cal B}|\leqslant |{\cal A}|\cdot |X|\cdot |Y|$.
\qed
\end{proof}

\section{Crisp-deterministic weighted finite automata\\ with output}

Let ${\cal A}=(A,X,Y,\sigma,\delta,\omega)$ be a Mealy-type  weighted finite
au\-tom\-aton over a semiring $S$.~The weighted transition function $\delta$ is called {\it crisp-determin\-is\-tic\/} if for all $x\in X$ and $a\in A$ there exists $a'\in A$ such that $\delta_x(a,a') = 1$, and $\delta_x(a,b) = 0$, for all $b\in A\setminus\{a'\}$.~Also,~the initial weight vector $\sigma$ is
{\it crisp-deterministic\/} if there exists $a_{0} \in A$ such that $\sigma(a_0) = 1$, and $\sigma(a) = 0$ for~all $a \in A\setminus\{a_0\}$. If both $\sigma$ and~$\delta$~are~crisp-deterministic, then~${\cal A}$ is called  a \emph{crisp-deterministic Mealy-type} \emph{weighted automaton}.

Equivalently, we define a crisp-deterministic~Mealy-type  weighted autom\-aton over a semiring $S$ as a tuple ${\cal A}=(A,X,Y,a_{0},\delta,\omega)$, where $A$ is~a~non-empty set of states, $a_{0}\in A$ is an initial state, $\delta: A \times X\to A$ is a transition function and $\omega: A\times X\times Y\to S$  is a weighted output function.~For any $x\in X$ we define $\delta_{x}:A\to A$ by $\delta_{x}(a)=\delta(a,x)$, for all $a\in A$, and for any $u\in X^*$ we define the transition function $\delta_{u}:A\to A$ as follows:~For any $a\in A$ we set $\delta_{\varepsilon}(a)=a$, and for $a\in A$, $u\in X^*$ and $x\in X$, we set $\delta_{ux}(a)=\delta_{x}(\delta_{u}(a))$.

A {\it crisp-deterministic Moore-type weighted automaton\/} over
$S$ is defined as a tuple ${\cal A}=(A,X,Y,a_{0},\delta,\omega)$, where everything
is the same as in the definition of a crisp-deterministic Mealy-type weighted automaton except the weighted output function, for which we assume~that $\omega: A\times  Y\to S$.

Given that crisp-deterministic Mealy-type weighted automata are a special type of the general Mealy-type  weight\-ed automata, the $1n$-semantics,
$n1$-semantics and sequential semantics for these automata are those that
are defined in Section \ref{sec:Mealy}.~Similarly, the definitions of $1n$-se\-man\-tics
and $n1$-se\-mantics for Moore-type weighted automata given in Section \ref{sec:Moore} apply also to crisp-determi\-nistic Moore-type weighted autom\-ata. However, in the case of crisp-deterministic Mealy-type and Moore-type weighted automata
 it is natural to consider the following semantics for which we prove that they are equivalent to all the above listed semantics.

\begin{definition}[Crisp-deterministic semantics]
The {\it cd-behavior\/} of a crisp-deterministic Mealy-type weighted automaton
${\cal A}=(A,X,Y,a_{0},\delta,\omega)$ is the function $\lBrack {\cal A}\rBrack_{cd}\!:(X\times Y)^*\to S$ defined by
\begin{equation}\label{eq:beh.cd.Mealy.e}
\lBrack {\cal A}\rBrack_{cd}(\varepsilon,\varepsilon)=1,
\end{equation}
and
\begin{align}\label{eq:beh.cd.Mealy.uv}
&\lBrack {\cal A}\rBrack_{cd}(u,v) = \\
&\hspace{3mm}=\omega_{x_{1},y_{1}}(a_{0})\cdot\omega_{x_{2},y_{2}}(\delta_{x_{1}}(a_{0}))\cdot
\ldots\cdot\omega_{x_{n},y_{n}}(\delta_{x_{1}...x_{n-1}}(a_{0})),\notag
\end{align}
for all $u=x_{1}x_{2}...x_{n}\in X^*$ and  $v=y_{1}y_{2}...y_{n}\in Y^*$.

Similarly, by the {\it cd-behavior\/} of a crisp-deterministic Moore-type weighted automaton
${\cal A}=(A,X,Y,a_{0},\delta,\omega)$ we mean the function $\lBrack {\cal A}\rBrack_{cd}\!:(X\times Y)^*\to S$ defined by
\begin{equation}\label{eq:beh.cd.Moore.e}
\lBrack {\cal A}\rBrack_{cd}(\varepsilon,\varepsilon)=1,
\end{equation}
and
\begin{align}\label{eq:beh.cd.Moore.uv}
&\lBrack {\cal A}\rBrack_{cd}(u,v) = \\
&\hspace{3mm}=\omega_{y_{1}}(\delta_{x_{1}}(a_{0}))\cdot\omega_{y_{2}}(\delta_{x_{1}x_2}(a_{0}))\cdot
\ldots\cdot\omega_{y_{n}}(\delta_{x_{1}...x_{n}}(a_{0})),\notag
\end{align}
for all $u=x_{1}x_{2}...x_{n}\in X^*$ and  $v=y_{1}y_{2}...y_{n}\in Y^*$.
\end{definition}

Now we show that for all crisp-deterministic Mealy-type and Moore-type weighted
automata the above defined semantics coincide
with all semantics defined in Sections \ref{sec:Mealy} and \ref{sec:Moore}
for the general Mealy-type and Moore-type weighted
automata.

\begin{theorem}\label{th:beh.cd.Mealy}
If ${\cal A}=(A,X,Y,a_{0},\delta,\omega)$ is a crisp-determinis\-tic Mealy-type weighted automaton then
\begin{equation}\label{eq:cd.sem.Mealy}
\lBrack {\cal A}\rBrack_{cd}= \lBrack {\cal A}\rBrack_{1n} = \lBrack {\cal A}\rBrack_{n1} = \lBrack {\cal A}\rBrack_{s},
\end{equation}
and if $\cal A$ is a crisp-determinis\-tic Moore-type weighted autom\-aton then
\begin{equation}\label{eq:cd.sem.Moore}
\lBrack {\cal A}\rBrack_{cd}= \lBrack {\cal A}\rBrack_{1n} = \lBrack {\cal A}\rBrack_{n1} .
\end{equation}
\end{theorem}

\proof
If ${\cal A}$ is a crisp-determinis\-tic Mealy-type weighted automaton, it is easy to check that the rightmost terms in equations (\ref{eq:1n-beh.Mealy.e}), (\ref{eq:n1-beh.Mealy.e}) and (\ref{eq:s-beh.Mealy.e})
become equal to 1,~while the rightmost terms in equations (\ref{eq:1n-beh.Mealy}), (\ref{eq:n1-beh.Mealy}) and (\ref{eq:s-beh.Mealy}) are converted into the term on the right-hand side of~equation (\ref{eq:beh.cd.Mealy.uv}).~Therefore,
(\ref{eq:cd.sem.Mealy}) holds.

Similarly, if ${\cal A}$ is a crisp-determinis\-tic Moore-type weighted automaton, then the rightmost terms in equations (\ref{eq:1n-beh.Moore.e}) and (\ref{eq:n1-beh.Moore.e})
are equal to 1,~whereas the rightmost terms in  (\ref{eq:1n-beh.Moore}) and (\ref{eq:1n-beh.Moore}) are transformed into the term on the right-hand side of (\ref{eq:beh.cd.Moore.uv}).~Thus, we conclude that (\ref{eq:cd.sem.Moore})
is true.
\qed

\end{document}